\documentclass[DIV=11]{scrartcl}

\usepackage{amsfonts,amsmath,amsthm,mathtools}
\usepackage{lmodern,microtype,setspace}

\usepackage{tikz}
\usetikzlibrary{calc}
\usetikzlibrary{quantikz}

\usepackage[colorlinks=false,breaklinks=true]{hyperref}

\usepackage[backend=bibtex8,style=alphabetic,doi=false,isbn=false,url=false,maxbibnames=5,sorting=nty,sortcites=false]{biblatex}

\newbibmacro{string+doi}[1]{\iffieldundef{doi}{#1}{\href{https://dx.doi.org/\thefield{doi}}{#1}}}
\DeclareFieldFormat{title}{\usebibmacro{string+doi}{\mkbibemph{#1}}}
\DeclareFieldFormat[article]{title}{\usebibmacro{string+doi}{\mkbibquote{#1}}}
\DeclareFieldFormat[incollection]{title}{\usebibmacro{string+doi}{\mkbibquote{#1}}}
\DeclareFieldFormat[inproceedings]{title}{\usebibmacro{string+doi}{\mkbibquote{#1}}}
\renewbibmacro{in:}{}

\usepackage[capitalize]{cleveref}
\crefname{lemma}{Lemma}{lemmas}
\crefname{proposition}{Proposition}{propositions}
\crefname{definition}{Definition}{definitions}
\crefname{theorem}{Theorem}{theorems}
\crefname{conjecture}{Conjecture}{conjectures}
\crefname{corollary}{Corollary}{corollaries}
\crefname{example}{Example}{examples}
\crefname{section}{Section}{sections}
\crefname{appendix}{Appendix}{appendices}
\crefname{figure}{Fig.}{figs.}
\crefname{equation}{Eq.}{eqs.}
\crefname{table}{Table}{tables}
\crefname{item}{Property}{properties}
\crefname{remark}{Remark}{remarks}
\crefname{problem}{Problem}{problems}

\newtheorem{theorem}{Theorem}
\newtheorem{definition}[theorem]{Definition}

\newtheorem{lemma}[theorem]{Lemma}

\input{braket}    
\input{problem}  

\usepackage{xspace}
\usepackage{soul}

\DeclareMathOperator{\Tr}{Tr}
\DeclareMathOperator{\poly}{poly}

\newcommand\ee{\mathrm e}
\newcommand\ii{\mathrm i}
\newcommand\dd{\mathrm d}
\newcommand\cost{\mathcal T}

\newcommand\EighProb{Randomized Quantum Eigenvalue Estimation Problem\xspace}
\newcommand{\eighprob}{\nameref{prob:sQEEA}\xspace}

\DeclareMathOperator{\BigO}{O}

\DeclareMathOperator{\Expectation}{\mathbb E}

\bibliography{bibliography}

\usepackage{authblk}

\author[a,b]{Laura Clinton}
\author[a,c]{Johannes Bausch}
\author[a]{Joel Klassen}
\author[a]{Toby Cubitt}
\affil[a]{Phasecraft Ltd.}
\affil[b]{Department of Computer Science, University College London}
\affil[c]{CQIF, DAMTP, University of Cambridge}

\title{Phase Estimation of Local Hamiltonians\\on NISQ Hardware}

\begin{document}
\maketitle

\begin{abstract}\setstretch{1.0}
  In this work we investigate a binned version of Quantum Phase Estimation (QPE) set out by \cite{Somma2019} and known as the Quantum Eigenvalue Estimation Problem (\nameref{prob:qeep}).
  Specifically, we determine whether the circuit decomposition techniques we set out in previous work, \cite{Clinton2020}, can improve the performance of \nameref{prob:qeep} in the NISQ regime.
  To this end we adopt a  physically motivated abstraction of NISQ device capabilities as in \cite{Clinton2020}.
  Within this framework, we find that our techniques reduce the threshold at which it becomes possible to perform the minimum two-bin instance of this algorithm by an order of magnitude.
  This is for the specific example of a two dimensional spin Fermi-Hubbard model.
  For example, given a $10\%$ acceptable error on a $3\times 3$ spin Fermi-Hubbard model, with a depolarizing noise rate of $10^{-6}$, we find that the phase estimation protocol of \cite{Somma2019} could be performed with a bin width of approximately $1/9$ the total spectral range at the circuit depth where traditional gate synthesis methods would yield a bin width that covers the entire spectral range.
  We explore possible modifications to this protocol  and propose an application, which we call \EighProb  (\eighprob).
  \eighprob outputs estimates on the fraction of eigenvalues which lie within randomly chosen bins and upper bounds the total deviation of these estimates from the true values. One use case we envision for this algorithm is resolving density of states features of local Hamiltonians, such as spectral gaps. 

\end{abstract}

\newpage
\tableofcontents
\newpage

\section{Introduction}

One of the most important subroutines in quantum computing is the Quantum Phase Estimation algorithm. In 1994 Shor showed that a quantum computer can factor an integer in a number of steps polynomially increasing with the bit size of that integer---a task as yet not known to be possible on a classical computer \cite{Shor}.
A little later, Kitaev~\cite{Kitaev1995} came up with an alternative formulation of the factoring algorithm. A key subroutine in Kitaev's formulation of the factoring algorithm is Quantum Phase Estimation (QPE).
QPE retrieves information about the eigenvalues of a unitary $U$ by inducing a phase kickback---through repeated controlled applications of $U$---onto a readout register, to which a quantum Fourier transform (QFT) is then applied.
In addition to the factoring algorithm, QPE appears as an essential subroutine in many other quantum algorithms; indeed it has been shown that QPE can be effectively applied to any problem efficiently solvable on a quantum computer \cite{Wocjan2006}.

A task for which QPE is naturally suited is the retrieval of spectral information about local Hamiltonians, wherein $U$ corresponds to time evolution under a given Hamiltonian and we are obtaining information about the energy eigenvalues of $H$.
This spectral information may be readily applied to the study of natural physical systems, for example in characterizing the properties of materials \cite{Kittel2004}.

Given the clear value of QPE it is worth understanding to what extent the subroutine may be implemented on Noisy Intermediate Scale Quantum (NISQ) devices.
In this context, minimizing both the number of qubits, and the circuit depth, are essential.
Two features of QPE which contribute significantly to these costs are the application of a large coherent QFT, and the number of times a controlled unitary is applied.
This motivates considering a modified form of QPE wherein the QFT is replaced with classical post processing, and one attempts to retrieve as much spectral information as possible under a limited number of applications of the controlled unitary.
This is the strategy employed recently in works by \citeauthor{Obrien2019} \cite{Obrien2019}, \citeauthor{Somma2019}  \cite{Somma2019}, and \citeauthor{lin2021heisenberg} \cite{lin2021heisenberg} in investigating applications of QPE for NISQ devices. In this work we focus in particular on the Quantum Eigenvalue Estimation Problem (\nameref{prob:qeep}) proposed by Somma.

The limiting factor for the performance of these methods is the maximum time $T$ for which one may perform the controlled unitary time evolution.
In previous work \cite{Clinton2020} it was proposed that one could leverage variable time two-qubit interactions at the hardware level to perform more efficient multi-qubit gate synthesis.
Under a cost model wherein the duration of the hardware interaction is proportional to the time parameter of the two-qubit interaction, this method can yield significant improvements on the total wall-time cost of Hamiltonian simulation.
Clearly such gains would have an impact on the performance of the phase estimation methods outlined above. 
In this paper we build on this work and consider the near-term feasibility of performing phase estimation on local Hamiltonians under this framework.

To do this we combine our gate synthesis techniques with \nameref{prob:qeep}.
We find that our techniques significantly reduce the threshold at which it becomes possible to perform \nameref{prob:qeep} for a two dimensional spin Fermi-Hubbard model.
For example, given a $10\%$ acceptable error on a $3\times 3$ spin Fermi-Hubbard model, with a depolarizing noise rate of $10^{-6}$, we find that \nameref{prob:qeep} could be performed with a bin width of approximately $1/9$ the total spectral range at the circuit depth where traditional gate synthesis methods would yield a bin width that covers the entire spectral range.
Furthermore, we reformulate \nameref{prob:qeep} as a general protocol to sample from a convolution between the spectrum of a Hamiltonian and a characterised class of functions. We then explore possible modifications to this protocol and propose an application, which we dub  \EighProb (\eighprob), and explain how this may be used to estimate spectral properties such as density of states.

The outline of this paper is as follows.
In \cref{sec:prelim} we discuss some important preliminary details.
These include: a review of our circuit error model, with a discussion of hardware and noise assumptions (\cref{subsec:error-models});
a review of QPE with and without QFT (\cref{subsec:phaseEst-by-cpp});
and a review of the Quantum Eigenvalue Estimation Problem (\cref{subsec:qeep}).
In \cref{sec:generic-QEEA} we
give a rigorous formulation of \eighprob, and show how this can be reduced to \nameref{prob:qeep}.
In \cref{sec:subcircuit-QEEA} we discuss the details of applying sub-circuit methods to \nameref{prob:qeep} and perform a comprehensive numerical analysis of the expected run-time of \nameref{prob:qeep} using these methods.
We illustrate how \nameref{prob:qeep} may be adapted to specific purposes by considering a modified version with a different smoothing function on the bins (\cref{subsec:altInd}).
This modified version yields smaller bin sizes under a max time evolution $T$, when $T$ is sufficiently small, but with a trade-off of worse asymptotic behaviour as $T$ increases.
Trade-offs of this type are particularly relevant in NISQ applications where minimising the circuit depth of the concrete problem instance is more important than optimal asymptotic scaling.
The numerical analysis is performed for the time evolution on a spin-full Fermi-Hubbard model on a $3 \times 3$ and $5 \times 5$ grid (\cref{subsec:numerics}).
For each lattice size we analyse the performance with target precisions of $10\%$, $5\%$ and $1\%$.
We do not assume perfect hardware, but in the analysis we model the noise inherent to all NISQ hardware.

We find that our decomposition methods consistently reduce the threshold at which it becomes possible to perform a minimum working example of \nameref{prob:qeep} with two bins by an order of magnitude over the range of lattice sizes and target precisions considered. 
This suggests that one or two orders of magnitude improvement in quantum hardware are required before \nameref{prob:qeep} enters the scope of practical application on near term devices.
In particular, the depolarizing noise rates we consider are extremely ambitious, and even with these noise rates we expect smaller bin widths are necessary for practically useful applications. 

\section{Preliminaries}\label{sec:prelim}

\subsection{Circuit Error Model}\label{subsec:error-models}

Here we recap our theoretical abstraction of the hardware capabilities of NISQ-era devices \cite{Clinton2020}. We describe three related concepts in this section: a circuit model, an error model, and a cost model. These concepts are summized by \cref{def:short-pulse-circuit,eq:noise-channel,def:time-cost}.
The circuit model can describe any quantum computing architecture in which discrete gates are generated by continuous-time interactions between qubits, and where these interactions can effectively be described by a Hamiltonian---or effective Hamiltonian---coupling those qubits.

Many current matter-based quantum computing architecture fit this description, including superconducting circuits, most ion trap architectures and NMR quantum computing (liquid-state or solid-state). Although this error model is a theoretical idealisation work has recently been published which supports the validity of this model and demonstrates the advantages of exploiting analogue control \cite{IBMshortpulse}. 

We call this circuit model the subcircuit model is defined as follows.
\begin{definition}[Sub-Circuit Model]\label{def:short-pulse-circuit}
	Given a set of qubits $Q$, a set $I \subseteq Q \times Q$ specifying which pairs of qubits may interact, a fixed two qubit interaction Hamiltonian $h$, and a minimum switching time $t_{\text{min}}$, a sub-circuit pulse-sequence $C$ is a quantum circuit of $L$ pairs of alternating layer types $C = \prod_l^L  U_l  V_l $ with $U_l = \prod_{i \in Q} u_i^l$ being a layer of arbitrary single qubit unitary gates, and  $V_l = \prod_{ ij  \in \Gamma_l}  v_{ij} \left(t_{ij}^l \right)$ being a layer of non-overlapping, variable time, two-qubit unitary gates: $$ v_{ij}(t)=\ee^{\ii t  h_{ij}}$$
	with the set $\Gamma_l \subseteq I$ containing no overlapping pairs of qubits, and $t \geq t_{min}$. Throughout this paper we assume $h_{ij} = Z_iZ_j$. As all $\sigma_i \sigma_j$ are equivalent to $Z_i Z_j$ up to single qubit rotations this can be left implicit and so we take  $h_{ij} =\sigma_i \sigma_j$.
\end{definition}
This defines the gate-set we have access to. We also refer to these gates as ``sub-circuit'' as  they exist one level of abstraction below the digital gate-set of the standard circuit model.
We note that \emph{any} standard quantum circuit can be expressed in this ``sub-circuit'' form as up to single qubit rotations a CNOT gate is equivalent to $\ee^{-\ii \frac{\pi}{4}ZZ}$.
This means that we can cost standard circuits in our framework as sub-circuit gates with $t = \BigO(1)$. We will use this fact in \Cref{subsec:numerics}.
We will express all quantum circuits in terms of sub-circuit gates throughout this work.
We do this so that we can cost all circuits under the same noise model.

The sub-circuit model is useful as it can distinguish circuits under an error model where errors accumulate proportionally to the time for which interactions are switched on. This is in contrast to traditional error models, where errors accumulate according to gate depth and all gates are costed equally.
This leads us to the following error model.
We will assume that errors accumulate proportionally to the time interactions in the system are switched on for.
Formally, we assume that for $Q$ qubits every layer $l$ of sub-circuit gates in a circuit $C=\prod_{l}U_l V_l$ is interspersed with a single qubit depolarizing noise channel $\mathcal{E}$. For each qubit in layer $l$ the full channel $\mathcal{N}_l$ is then approximated as
\begin{align}\label{eq:noise-channel}
	\mathcal{N}_l(\rho): \rho \rightarrow (1-q|t^l|) \rho + q |t^l| \frac{I}{2^Q},
\end{align}
for some noise parameter $q \in [0,1]$ and where we have upper bounded each individual pulse-time in the $l$\textsuperscript{th}-layer as $|t^l|\coloneqq \text{max}_{ij \in \Gamma_l}\left(|t^l_{ij}|\right)$.

Finally, we define the following cost model. We have assumed that errors accumulate based on the length of time two-qubit interactions are switched on for and therefore we cost circuits according to the following definition.
\begin{definition}[Sub-Circuit Cost] \label{def:time-cost}
	The physical run-time of a sub-circuit pulse-sequence $C$ is defined as
	$$
	\cost(C) \coloneqq \sum_l^L \max_{ij \in \Gamma_l}\left(t_{ij}^l\right)
	$$
\end{definition}
The run-time is normalised to the physical interaction strength, so that $\vert h\vert = 1$. This cost metric assumes that the single qubit layers contribute a negligible amount to the total time duration of the circuit.

This cost-metric will help us estimate what circuits can and can't be implemented for a given depolarizing noise rate $q$ if we are willing to tolerate a target error $\epsilon_{\text{tar}}$.

\newcommand{\Ucirc}{\op U_\mathrm{circ}}
\let\O\relax
\newcommand{\O}{\BigO}
\DeclareDocumentCommand{\gateA}{m m}{%
	\draw[fill=white] (#1,#2-0.15) rectangle (#1+0.5,#2+1.15);
}
\DeclareDocumentCommand{\gateRow}{m m O{2} O{7}}{%
	\foreach \x in {0,...,#4}{
		\gateA{\x*#3+#1}{#2}
	}
}
\DeclareDocumentCommand{\errorRow}{m m O{2}}{%
	\foreach \x in {0,...,8}{%
		\draw[fill=white] (\x*#3+#1,#2) circle[radius=.35] node[] {$\mathcal E$};
	}
}

This is captured by the following. We can calculate a trivial error bound by considering the probability that no error occurs at all. Let $\mathcal{U}(\rho)$ be the channel which implements circuit $C$ without noise, and $\mathcal{U'}(\rho)$ the channel which implements circuit $C$ with noise. Then for any measured observable $\BigO$ we get
\begin{align}
\epsilon := \Tr \left[\left(\mathcal{U}(\rho)-\mathcal{U'}(\rho)\right)\BigO\right] \leq \epsilon_{\text{tar}},
\end{align}
as we want to determine the the maximum allowable run-time which keeps this stochastic error $\epsilon$ below a target error $\epsilon_{\text{tar}}$.

This stochastic error $\epsilon$ is upper bounded by the probability that any error occurs in the circuit. For a $Q$-qubit circuit the noise channel defined above this gives us
\begin{align}
\epsilon &\leq 1 - \prod_{l}(1-q |t^l|)^Q \\
&\leq 1 - (1-q)^{V} \leq \epsilon_{\text{tar}} ,
\end{align}
where $V = Q \times \mathcal{T}(C)$ is the volume of the circuit and $\mathcal{T}(C) = \sum_{l} |t^l|$. Inverting this, we say that given an error parameter $q \in [0,1]$, we can only implement circuits whose run-time is bounded as
\begin{align}\label{eq:stoch-noise-bound}
\mathcal{T}(C) \leq \frac{\log(1-\epsilon_{\text{tar}})}{Q \log(1-q)},
\end{align}
assuming we are willing to tolerate a stochastic error rate at most $\epsilon_{tar}$. We will use this in \Cref{subsec:numerics} and particularly in \Cref{fig:main}. These are the key assumptions we make in our circuit error model. In summary, we assume there is sufficient analogue control to express circuits in the sub-circuit form of \Cref{def:short-pulse-circuit} and that depolarizing noise in the device can be described by \Cref{eq:noise-channel}. This in turn makes the cost-metric of \Cref{def:time-cost} useful as then the trivial stochastic error bound \Cref{eq:stoch-noise-bound} can help us estimate requisite depolarizing error rates $q$ for a given task. Our analysis and benchmarking in \Cref{sec:subcircuit-QEEA} will follow this logic.

Having established this circuit error model in our previous work \cite{Clinton2020}, we also derived novel exact gate synthesis techniques to minimize the sub-circuit cost $\cost$ in this error model. These results can be summarized as follows. For a $k$-local Pauli interaction $P_k$, there exists a sub-circuit pulse-sequence  $C\coloneqq\prod_l^{L} U_l V_l$ which implements the evolution operator $\ee^{-\ii t  P_k}$. Most importantly, for any target time $t$ the run-time of that circuit is bounded as
\begin{align}\label{eq:general-k-cost}
	\cost \left(C \right) \leq \BigO\left( |t|^{\frac{1}{k-1}}\right),
\end{align}
according to the notion of run-time established in \cref{def:time-cost}. The details of the decompositions which accomplish this can be found in our previous work \cite[Appendix B]{Clinton2020}. Here we only note that they are exact and minimize $\cost(C)$ as $t \rightarrow 0$ at the expense of increased circuit depth.

\subsection{Quantum Phase Estimation By Classical Post Processing}\label{subsec:phaseEst-by-cpp}
Let $f(t)$ be an integrable function. For consistency with past literature we adopt the convention for the Fourier transform and its inverse to be
\begin{align}
    F(w) = \mathcal F[f](w) = \frac{1}{2\pi} \int_{\mathbb{R}} f(t) \ee^{-\ii w t} \dd t,
    \quad\text{and}\quad
    \mathcal F^{-1}[F](t) = \int_{\mathbb{R}} F(w) \ee^{\ii w t} \dd w.
\end{align}
For a periodic function $f$ one can express $f(t)$ as a Fourier series via
\begin{align}
    f(t) = \sum_{w=-\infty}^\infty F_w \ee^{\ii w t}
    \quad\text{with coefficients}\quad
    F_w \coloneqq F(w).
\end{align}
With these conventions, the Fourier inversion theorem reads $\mathcal F^{-1} = \mathcal F\mathcal R$ for the reflection functional $R[f](t) \coloneqq f(-t)$.

\emph{Quantum Phase Estimation} (QPE) describes a family of quantum algorithms for retrieving information about the phases $\phi_j$ associated with the eigenvalues $\ee^{\ii \phi_j}$ of a unitary operator $U$.
Such a unitary operator is generated by a hermitian operator $H$ via $U=\exp(-i H) = \sum_{n} \exp(-i \lambda_n) \ketbra{\phi_n}$, where $\{ \lambda_n \}$ are the eigenvalues of $H$---which are equivalent to the phases $\phi_j$ of $U$ modulo $2\pi$---and where the $\{\ket{\phi_n} \}$ are the (same) eigenvectors of $H$ and $U$, respectively.

All variants of phase estimation may be thought of in the following terms. Consider some input state $\rho$ and a time evolution operator $U(t)$ with phases $\{\phi_j\}$ and eigenstates $\{\ket{\phi_j}\}$. Prepare some time series data
\begin{equation}\label{eq:g(t)}
g(t) \coloneqq \sum_n \ee^{\ii \phi_n t} \bra{\phi_n} \rho \ket{\phi_n},
\end{equation}
and extract information about the spectrum of this data
\begin{equation}
G(w)\coloneqq \mathcal{F}[g] = \sum_i \delta(w-\phi_i)  \bra{\phi_i} \rho \ket{\phi_i}.
\end{equation}
In the original phase estimation protocol---illustrated in \cref{fig:QPE}---this process is performed coherently.
First one prepares the desired input state $\rho=\rho_I$ on an input register $I$, and a coherent superposition of time states $\sum_{t=1}^T\ket t$ on a control register $C$.
The time register then controls the unitary time evolution on register $I$, to obtain
\begin{equation}
\rho_{CI} \propto \sum_{tt'} \ketbra{t}{t'}_C \otimes U(t) \rho_I U(t')^\dagger.
\end{equation}
Tracing out the input register then yields time series data, namely
\begin{equation}
\Tr_I(\rho_{CI}) \eqqcolon \rho_C \propto \sum_n \sum_{tt'} \ketbra{t}{t'}_C \ee^{-\ii \phi_n t} \bra{\phi_n}\rho_I\ket{\phi_n} \ee^{\ii \phi_n t'} = \sum_{t t'}g(t'-t) \ketbra{t}{t'}_C .
\end{equation}
The quantum Fourier transform is applied, and we obtain
\begin{equation}
\textrm{QFT}(\rho_C) \propto \sum_{w w'}\left(\sum_{t t'} \ee^{-\ii w t + \ii w't'} g(t'-t)\right) \ketbra{w}{w'}_C.
\end{equation}
Measuring the control qubits in the computational basis is equivalent to sampling from $G(w)$ over possible $n$-bit values of $w$.

\begin{figure}
\begin{minipage}{.3 \linewidth}
\centering
\begin{quantikz}
\lstick{$\ket{+}$} &\ctrl{1} & \meter{$X+iY$} \arrow[r] & \rstick{$g(t)$} \\
\lstick{$\rho$} & \gate{U(t)} & \qw &
\end{quantikz}
\end{minipage}
\hspace{2cm}
\begin{minipage}{.5\linewidth}
\centering
\begin{quantikz}
\lstick{$\sum_{t=1}^{T} \ket{t}$} &\ctrl{1} &\gate{QFT}& \meter{$Z$} \arrow[r] & \rstick{$G(w)$} \\
\lstick{$\rho$} & \gate{U(t)} & \qw &
\end{quantikz}
\end{minipage}
\caption{Generic schematics for QPE with a QFT (right) and without (left). The circuit on the left describes how to obtain the time series data $g(t)$ from \cref{eq:g(t)}.
The controlled operation in the right circuit is a time control, in the sense that conditioned on the (multi-qubit) control lane being in state $\ket t$, the unitary $U(t)$ is executed.
}\label{fig:QPE}
\end{figure}
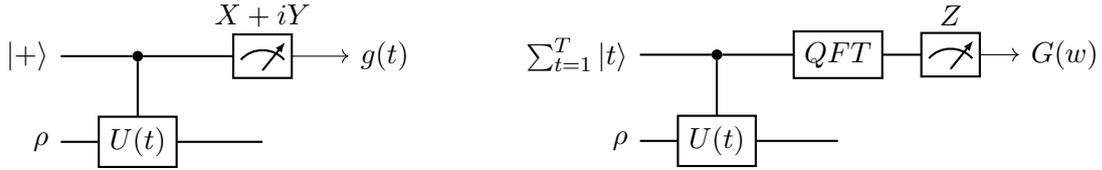

Applying QPE to problems involving local Hamiltonians differs significantly from its application to factoring in two important ways.
First, it is unlikely that one can construct a polynomial depth circuit decomposition
of $U(T)$ for exponentially large $T$ \cite{atia2017fast}.

Second, the size of a problem involving a local Hamiltonian goes as the number of qubits.
As the number of qubits grows, the number of eigenvalues increase exponentially, but the spectral range of the Hamiltonian increases polynomially \footnote{The spectral norm is upper bounded by the sum of the norm of the local terms.}. 
It is the spectral range of the Hamiltonian which determines the bandwidth of $g(t)$, and so it is not clear that the inclusion of a QFT should yield significant gains for such problems.
If on the other hand one wishes to resolve fine grain details of $G(w)$, then from a classical signal processing perspective the parameter which dictates this resolution is the maximum simulated time $T$, not the number of samples of $g(t)$.

This motivates an alternative method, which avoids the use of an expensive QFT.
Instead of preparing and processing a coherent superposition of $g(t)$ sampled at an exponential number of time steps, $g(t)$ is directly measured for a polynomial number of time steps up to a maximum time $T$, and this data is processed classically.
The procedure for measuring $g(t)$ (illustrated in \cref{fig:QPE}) is a straightforward phase kickback protocol, where the control register is prepared in the $\ket{+}$ state, the input register is prepared with the input state $\rho_I$, and a controlled time evolution is applied
\begin{equation}
\rho_{CI} = \frac{1}{2} \sum_{ij=0}^1 \ketbra{i}{j} \otimes U(t)^i \rho_I U(-t)^j
\end{equation}
Measuring the expectation value of $X+iY$ on the control register yields:
\begin{align}
 \langle X_C \rangle + i \langle Y_C \rangle  = \Tr[\rho_I U(t)]
=  \sum_i \ee^{- i  \phi_i t} \bra{\phi_i}\rho_I \ket{\phi_i} = g^*(t) = g(-t). \label{eq:gk}
\end{align}
Equipped with a quantum black box for values of $g(t)$  it would appear that the remainder of the problem is merely a classical signal processing problem, about which a rich body of literature exists. However the existing classical signal processing literature typically assumes that the signal is composed of a polynomial number of frequencies, which is not generally true for a signal as in \cref{eq:g(t)} composed of up to exponentially many distinct eigenvalues of the Hamiltonian of interest.

\subsection{The Quantum Eigenvalue Estimation Problem and Algorithm}\label{subsec:qeep}
As outlined in the last section, resolving all exponentially many eigenvalues of a Hamiltonian is costly, as it would require us to simulate an exponential time evolution under $H$.
Yet it is conceivable that one can obtain approximate information about the spectrum in a more efficient manner. Such an algorithm was introduced by \cite{Somma2019} and we recap it here.
The \emph{quantum eigenvalue estimation problem} (\nameref{prob:qeep})---as defined by \cite{Somma2019}---is just such an approximation.
In contrast to the ambitious goal of obtaining exponentially many eigenvalues for an operator $H$, the range of possible eigenvalues is divided up into bins, and the aim is to estimate
the fraction of eigenvalues in each bin up to some smoothing on the boundaries of the bin.

\begin{problem}[Quantum Eigenvalue Estimation Problem][QEEP]\label{prob:qeep}
\probleminput{
A Hamiltonian $H$ with associated eigenvalues $\{\lambda_i\}$ and eigenvectors $\{\ket{\lambda_i}\}$.
An input quantum state $\rho$.
A bin width $\eta>0$, precision parameter $\epsilon>0$ and confidence level $c<1$.\footnote{In Somma's original construction the bin width $\eta$ is equal to the precision parameter $\epsilon$. We have decoupled them, since they are conceptually distinct. }
An indicator function $f(w): \mathbb{R} \rightarrow [0,1]$ s.t. $\forall w \not\in [-\eta, \eta]: \;  f(w) =  0$ and $\forall w \in [0, \eta]: \; f(w)+f(w-\eta)=1$.
}
\problempromise{$||H||_{\infty} \leq 1/2$.}
\problemquestion{A vector $\textbf{q}=(q_0, q_1,..., q_M)$ with $M=\lfloor 1/\eta \rfloor$ s.t.\ with probability $\ge c$,  $||\textbf{q}-\textbf{p}||_1 \leq \epsilon$.
  Here
  $p_j = p(w_j)\coloneqq\sum_{i} f(\lambda_i-w_j) \bra{\lambda_i}\rho\ket{\lambda_i}$,
 
  and $w_j =j\eta-1/2$.
}
\end{problem}
More specifically if one sets the input state to be the maximally mixed state this becomes  $p_j = p(w_j)\coloneqq\sum_{i} f(\lambda_i-w_j) 2^{-Q}$, where $Q$ is the number of qubits.

The indicator function $f$ in \nameref{prob:qeep} 
 may be thought of as delineating the ``shape'' of a bin. The values $p_j$ can be thought of as the probability of observing the initial state $\rho$ within a certain energy band, centred around $j \eta - 1/2$, whose resolution and sharpness is determined by the indicator function $f$. Equivalently, one may understand the values $p_j$ as evenly spaced samples of $G(w)$ convolved with $f$:
\begin{equation}\label{eq:qeep-as-conv}
 p(w)=[G * f](w)
\end{equation}

Depending on the choice of $f$, solutions to a \nameref{prob:qeep} instance might have varying applications, and may be easier or harder to solve. A general scheme for constructing a function $f$ satisfying the conditions outlined in the problem statement is to take the ideal indicator function
\begin{equation}
\textrm{rect}_\eta(w) = \begin{cases}1 & w \in [-\eta/2, -\eta/2] \\ 0 & \textrm{otherwise} \end{cases}
\end{equation}
and convolve it with a normalized smoothing kernel $h(w) \in [0, 1]$ with support only in $[-\eta/2, \eta/2]$ \footnote{Demanding that $h(w)$ is normalized as $\int h(\omega) d \omega = 1$ ensures that $\forall w \in [0, \eta]: \; f(w)+f(w-\eta)=1$.} to obtain
\begin{equation}\label{eq:fConv}
f(w) = [\textrm{rect}_\eta * h](w).
\end{equation}

The algorithm that Somma proposes for resolving \nameref{prob:qeep} relies on the following signal processing principle. For a well chosen smoothing kernel $h$, the Fourier transform $P(t)\coloneqq\mathcal{F}[p](t)$ of $p(w)$ may be made to fall off rapidly beyond a maximum time $T$, by smoothing out the sharply varying spectral function $G(w)$. Furthermore, since $p(w)$ has bounded support, $P(t)$ need only be sampled at a finite rate. Thus, one need only sample $g(t)$ at a finite rate to a maximum time $T$ to retrieve a good approximation to the values $p_j$.

More precisely, we may consider that since $G(w)$ and $f(w)$ have bounded support they may be taken to be $2\pi$ periodic\footnote{Note that the period could be chosen tighter than $2\pi$ as $G(w)$ is bounded within $[-1/2,1/2]$, which will likely yield a small constant improvement in the overall cost. We will not analyse this here.}
and decomposed into Fourier series expansions:
\begin{align}\label{eq:r-g-pair}
G(w) =\frac{1}{2 \pi} \sum_{t = - \infty}^{\infty} g_{-t} \ee^{\ii t w} \quad \text{with} \quad g_t = g(t)
\end{align}
\begin{align}\label{eq:f-F-pair}
f(w) = \sum_{t = - \infty}^{\infty} F_t \ee^{\ii t w} \quad \text{with} \quad F_t = \mathcal{F}[f](t).
\end{align}
and thus:
\begin{equation}
p(w) = [G * f](w) = \int_{-\infty}^{\infty} G(w') f(w-w') dw' = \sum_{t = - \infty}^{\infty}\ee^{\ii t w} F_t  g_{-t}
\end{equation}
For fixed values $w_j = -1/2+j\eta$ we retrieve
\begin{equation}\label{eq:pExact}
p_j = \sum_{t = - \infty}^{\infty} F_t^j g_{-t} \;,\; F_t^j=\ee^{\ii t w_j} F_t
\end{equation}

In Somma's construction the indicator function $f(w)$ is defined as in \cref{eq:fConv}, with a smoothing kernel $h(w)= a \exp(-1/(1-(cw)^2) )$. With appropriate scaling (correct choice of $c$ and normalization $a \approx 2.25$). It is then argued that for any bin width $\eta$ there exists some constant $\alpha > 1$ such that for all $|t| \geq \alpha/\eta$
\begin{equation}\label{eq:Fdecay}
|F_t| \le \eta \exp\left(- \sqrt{|t| \eta/2}\right).
\end{equation}
Thus truncating the sum in \cref{eq:pExact} up to some maximum value $T \geq \alpha/\eta$ yields an approximation
\begin{equation}\label{eq:pApprox}
p_j' = \sum_{t = - T}^{T} F_t^j g_{-t}
\end{equation}
such that $\| \mathbf{p'} - \mathbf{p}\| \le \frac{\epsilon}{2}$.
We include an explicit computation of this maximum value $T$ in \cref{app:TBound}. However we emphasize that this is an asymptotic result, and that finding $\alpha$ is a non-trivial task. For this reason we also consider a different indicator function in \cref{subsec:altInd}.

In addition to the source of error from this truncation, there is a source of error stemming from imperfect measurements of $g(t)$. It is assumed that the $T$-dimensional vector of actual measurements $\tilde{\textbf{g}}$ satisfies
$\| \tilde{\textbf{g}}-\textbf{g} \|_1 \leq \epsilon/2 $ with probability greater than the confidence parameter $c$. Thus one retrieves a vector $\textbf{q}$ given by
 \begin{equation}\label{eq:qApprox}
q_j = \sum_{t = - T}^{T} F_t^j \tilde{g}_{-t}
\end{equation}
such that $\| \textbf{q}-\textbf{p} \|_1 \leq \epsilon $ with probability greater than $c$. Further details concerning this can be found in \cite{Somma2019}. Importantly $c$ depends on the number of repetitions of the circuit depicted in \cref{fig:QPE}; and not the maximum evolution time $T$.

Choosing the $1$-norm yields bounds on the \emph{overall} deviation from the ideal vector $\textbf{p}$. One may consider other norms, such as the max norm, which bounds the \emph{maximum} deviation, and yields a milder variant of Somma's error bound. For completeness and for the purpose of deriving explicit constants in the error bounds we walk through Somma's error derivation in detail in \cref{app:TBound}.

\section{\EighProb}\label{sec:generic-QEEA}
\subsection{Problem Definition}\label{subsec:generic-QEEA-def}
As discussed in the introduction, a fundamental numerical primitive is to ask for eigenvalues for a given normal or hermitian operator $A$.
More specifically, given a hermitian matrix $A$, the task is to return its spectrum---the set of its eigenvalues---up to some bit precision $b$, ie numerical accuracy $\xi=2^{-b}$.
Instead of speaking of the numeric precision with which results are obtained, one can take the viewpoint that such a function sorts $A$'s eigenvalues into $2^b$ many bins of width $\propto \xi$.
Such a function is e.g.\ implemented in matlab's \texttt{eig} or Mathematica's \texttt{Eigenvalues} functions.
Assuming no other numerical errors are present during the calculation,\footnote{Arbitrary precision arithmetic or internal arithmetic can give certain guarantees on the precision of the eigenvalues obtained.} we expect the result to be such that eigenvalues are rounded to their nearest bin.

As currently formulated, the Quantum Eigenvalue Estimation Problem (\nameref{prob:qeep}) contains overlapping indicator functions $f^j(\omega) := f(\omega - \omega_j)$, so that an eigenvalue in the support of both $f^j$ and $f^{j+1}$ will contribute to both $p_j$ and $p_{j+1}$. There is an important reason for formulating \nameref{prob:qeep} in this manner. Consider what a ``perfect'' choice of non-overlapping $f^j(\omega)$ might look like. As described in \cite{Somma2019}, the natural choice would be an indicator function such as
\begin{align}
	f(\omega) = \begin{cases}
		1 &  - \eta/2 \leq \omega <  + \eta/2 \\
		0 & \text{otherwise}.
	\end{cases}
\end{align} 
This choice would resolve the spectrum perfectly; and furthermore remove $f(\omega)$ from the input of the problem statement. However, as explained by the authors of \cite{Somma2019} the slow decay of the Fourier coefficients would mean one would have to simulate $\ee^{-\ii T H}$ for impractically large $T$. Due to this one must choose a smooth $f(\omega)$ which approximates this ideal solution through having smoothly decaying support over $ -\eta \leq \omega<  + \eta$. That is to say the indicator functions overlap each other so that $\forall w \in [0, \eta]: \; f(w)+f(w-\eta)=1$. 

To address these considerations we formulate a problem which we call the Randomized Quantum Eigenvalue Estimation Problem, or \eighprob for short, which is close in spirit to Matlab's eig or Mathematica's Eigenvalues function. For a given $H$ \eighprob outputs a binned estimate of the spectrum up to a probabilistic guarantee on the likelihood that eigenvalues close to bin boundaries have been misplaced; and which replaces the issue overlapping bins with \emph{randomly defined} bins. We will now define this problem formally and then propose an algorithm to solve \eighprob via a reduction to \nameref{prob:qeep}.

\begin{problem}[\EighProb][rQEEP] \label{prob:sQEEA}
\probleminput{
    Local Hamiltonian $H=\sum_i h_i$.
    Maximal point spacing $\xi > 0$, deviation $\Delta\ge0$, and confidence $c \in [0, 1]$.
}
\problempromise{$\operatorname{spec}(H) \subseteq [-1/2,1/2]$.}
\problemquestion{
    $m \coloneqq 2/\xi$. List of coordinates $(x_i, y_i)$ for $i=1, \ldots, (m+1)$  with $x_0 \coloneqq -1/2$ and $x_{m+1} \coloneqq 1/2$, such that
    \begin{enumerate}
    \item  $x_i$ are randomly but uniformly spaced, ie for all $x\in[-1/2,1/2]$ there exists a $x_i$ s.t.\ $|x-x_i| < \xi/4$, and  the distance between neighbouring points is bounded by $|x_i-x_{i-1}| < \xi$.
    \item the total absolute difference between the estimated ($y_i$) and true ($n_i$) number of eigenvalues of $H$ within the interval $B_i=[x_{i-1}, x_i)$ is upper bounded -- with probability greater than $c$ -- by the deviation $\Delta$, ie $\sum_i |y_i - n_i| \le \Delta$.
    \end{enumerate}}
\end{problem}
We first remark that the above problem is well-defined: indeed,
we can always first obtain the entire spectrum of $H$ to sufficient precision, and choose the abscissas such that the intervals $B_i$ never have an eigenvalue on their boundaries (which is possible even with the restriction of $n=\poly 1/\xi$, as the spectrum of $H$ is finite); and then define $y_i$ as the number of eigenvalues within that interval. This answers the problem even for deviation $\Delta=0$ and confidence $c=1$.

Naturally, what we will be interested in is whether we can answer \eighprob 
efficiently, and for which sets of parameters we can do so.

\subsection{Reducing \eighprob to QEEP}\label{subsec:Generic-QEEA-alg}
\newcommand{\lwr}{^\mathrm{lwr}}
\newcommand{\upr}{^\mathrm{upr}}
Given a local Hamiltonian $H=\sum_i h_i$ acting on $Q$ qubits (which we assume to be rescaled such that the spectrum lies within $[-1/2,1/2]$), we can solve the \eighprob by using \nameref{prob:qeep} as a subroutine and given access to a \nameref{prob:qeep} solver with a sub-circuit cost of $\cost(Q,\eta,\epsilon,c)$.  We do this via the following algorithm. 
\begin{figure}[t]\label{fig:rQEEP-envelopes}
	\centering
	\includegraphics[width=0.7\linewidth]{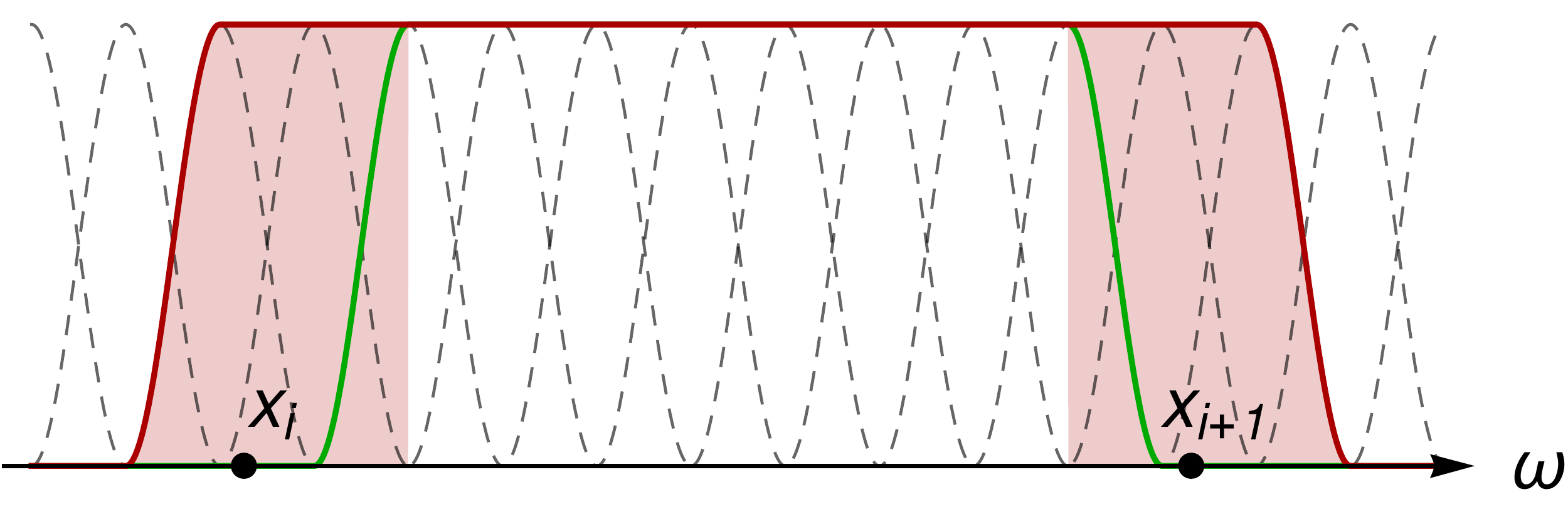}
	\caption{The dashed grey lines are QEEP indicator functions $f(\omega - \omega_j)$. The red envelope highlights which $f(\omega - \omega_j)$ have compact support in $R\upr_i$ and the green envelope those which have compact support in $R\lwr_i$. If an eigenvalue lies in the red shaded region then $y\lwr_i \neq y\upr_i$ and the estimate $y_i$ deviates from $n_i$, the true number of eigenvalues in bin $B_{i+1} := [x_i, x_{i+1})$. In \Cref{ap:rQEEP-proofs} we use a volume argument to bound the likelihood of these randomly chosen shaded regions falling on an eigenvalue of $H$ and thus prove \Cref{thm:runtime-main}.}
\end{figure}
\begin{enumerate}
	\item Given a maximal point spacing $\xi > 0$, deviation $\Delta\ge0$, and confidence $c \in [0, 1]$ uniformly partition the interval $[-1/2,1/2]$ into $m\coloneqq 2/\xi$ many segments 
	\begin{align}
		\left[ -\frac12, -\frac12 + \frac\xi2 \right) \cup \left[ -\frac12 + \frac\xi2, -\frac12 + \xi \right) \cup \ldots \cup \left[ \frac12 - \frac\xi2, \frac12 \right] \eqqcolon S_1 \cup \ldots \cup S_m.
	\end{align}

	\item With confidence $c$ run the following subroutine $\lceil \log_{(1-\frac{c}{2})}(1-c) \rceil$ times, and collect the results in a list:
	\begin{enumerate}
		\item Within each segment, choose $x_i$ uniformly at random from $S_i$. Set $x_0 \coloneqq -1/2$, and $x_{m+1} \coloneqq 1/2$.
		\item Solve a \nameref{prob:qeep} instance with a suitable indicator function $f$, input state $\rho = I/2^Q$, bin width $\eta \coloneqq \Delta/(6\times 2^{Q+1}(m+1)^2)$, precision $\epsilon=\Delta /(8 \times 2^{Q+1})$, and confidence $c$. Denote the output with $\mathbf q=(q_0,\ldots,q_M)$ where $M=\lfloor 1/\eta \rfloor$, for the bins centered around $w_j = -1/2 + j\eta$ for $j=0,\ldots,M$.
		\item For all $i=0,\ldots,m$, define lower and upper envelopes via
		\begin{align}
			y\lwr_i &\coloneqq 2^Q \!\!\!\! \sum_{w_j \in R\lwr_i} \!\!\!\! q_j
			&\text{and}&&
			y\upr_i &\coloneqq 2^Q \!\!\!\! \sum_{w_j \in R\upr_i} \!\!\!\! q_j\\
			R\lwr_i &\coloneqq [x_i + \eta, x_{i+1} - \eta ]
			&\text{and}&&
			R\upr_i &\coloneqq [x_i - \eta, x_{i+1} + \eta ]
		\end{align}
		and where we implicitly assume $R\lwr_i = \emptyset$ if $x_i + \eta > x_{i+1} - \eta$.
		\item Return $y_i \coloneqq (y\lwr_i + y\upr_i) / 2$.
	\end{enumerate}
	\item Return the result from the subloop with the smallest deviation of lower and upper bounds.
\end{enumerate}

In \Cref{ap:rQEEP-proofs} we prove the above algorithm has the following quantum overhead
\begin{theorem}\label{thm:runtime-main}
Given access to a \nameref{prob:qeep} solver with a subcircuit runtime $\cost(Q,\eta,\epsilon,c)$ then the following holds. For a point spacing $\xi > 0$ (or equivalently $m := 2/\xi$), deviation $\Delta\ge0$, and confidence $c$, the overall subcircuit runtime to answer \eighprob using \nameref{prob:qeep} is 
\begin{align}
	\cost\left(Q,\frac{1}{12 (m+1)^2} \times \frac{\Delta}{2^Q},\frac{1}{16} \times \frac{\Delta}{2^Q},c \right),
\end{align}
and the additional classical overhead is $\log_{(1-\frac{c}{2})}(1-c)$.
\end{theorem}
Using \cref{thm:runtime-main} we can solve for $m$ and $\Delta$ in terms of $\epsilon$, $\eta$ and $Q$ and find the following dependence
\begin{align}\label{eq:m}
		m &= \sqrt{\frac{4\epsilon}{3\eta}} -1 \\
		\Delta &= 2^Q \times 16 \epsilon.\label{eq:Delta}
\end{align}
\Cref{eq:m,eq:Delta} tells us the number of bins and deviation with which we can expect when solving an instance of \eighprob given access to a \nameref{prob:qeep} solver which can achieve and bin width of $\eta$ and a precision $\epsilon$.
The question is whether this now allows us to solve a problem of interest.
The total deviation $\Delta$ refers to maximum number of misplaced eigenvalues. From the above we see that the total \emph{fraction} of misplaced eigenvalues is proportional to $\epsilon$, even though the total number scales as $2^Q$.

The runtime $\cost$ depends on this $\eta$ and $\epsilon$.
If we assume that we can solve \nameref{prob:qeep} efficiently for $\eta, \epsilon = \text{poly}(Q^{-1})$, then we can expect to efficiently solve \eighprob for a polynomially scaling number of bins $m = \text{poly}(Q)$ and a polynomially scaling fraction of misplaced eigenvalues, $\Delta/2^Q = \text{poly}(Q^{-1})$.

From \cite{Brown2011} we know that the computational difficulty of calculating the density of states for local Hamiltonians is in the counting version of QMA (known as \#BQP) and so we we cannot use \eighprob to do this efficiently. However, we can still determine important features of the spectrum. The existence of a Mott insulator transition in Fermi-Hubbard models, and the Hamiltonian parameters for which this transition occurs, are of widespread interest \cite{review_mott}. As the Mott insulator phase features a spectral gap we could try to observe this with \eighprob.
Generally speaking we could  resolve gaps in the spectrum of a family of local Hamiltonians, provided this gap shrinks no faster than polynomially in the system size. This is because we can tolerate a polynomially increasing number of bins and so maintain the same resolution as we increase $Q$

\section{Implementing and Benchmarking QEEP and rQEEP for Fermi-Hubbard}\label{sec:subcircuit-QEEA}
\subsection{Sub-circuit Implementation with Local Controls}\label{subsec:subcicuit-soln-QEEP}

In this section we propose an alternative protocol to obtain $g_t$ using several geometrically local control ancillas in place of a single globally connected ancilla as described in \cite{Somma2019}.

We consider a reasonably general case. Suppose we wish to solve \nameref{prob:qeep} for a $k$-local qubit Hamiltonian, $H = \sum_{i} \alpha_i h_i$. As we are interested in the NISQ era algorithms we will assume that when laid out on hardware, every Pauli interaction $h_i$ acts non-trivially on geometrically local qubits. Furthermore we assume that the device has sufficient connectivity so that we can introduce a geometrically local ancilla qubit for each of these interactions. We denote these ancillas by dashed indices, $i'$, in contrast to the data qubits with undashed indices. This is illustrated for a $2$-local Hamiltonian with two terms in \cref{fig:local-time-series}.

Over the joint Hilbert space we define an augmented Hamiltonian
\begin{align}
H' \coloneqq \sum_{i=1} \alpha_i  Z_{i'} \otimes h_i .
\end{align}
To obtain time series data $g_t$ from this augmented Hamiltonian, we first prepare a cat state on the control ancillary space, $C$. For $\rho_I= \ketbra{\psi}_I$ the overall input is
\begin{align}
\ket{\psi(0)}_{CI} =  \frac{1}{\sqrt{2}} \big( \ket{00 \cdots 0}_C + \ket{11 \cdots 1}_C\big)\otimes \ket{\psi}_I.
\end{align}
On this initial state, the cat state induces a sign flip on the augmented Hamiltonian $H'$, conditioned on the cat state to be in state $\ket{1}_C$; more precisely, applying the evolution operator $U'_c(t) \coloneqq \exp(\ii t H'/2)$ we get
\begin{align}
	 \ket{\psi(t)}_{CI} &= \frac{1}{\sqrt{2}} \left( \ket{00 \cdots 0}_C \ee^{\ii t H/2} \ket{\psi}_I + \ket{11 \cdots 1}_C \ee^{- \ii t H/2} \ket{\psi}_I \right),
\end{align}
as can be readily verified.
By measuring the expectation value $g_t = \langle X_{1'} + \ii Y_{1'} \rangle_{\rho(t)}$ for $\rho(t) = \ketbra{\psi(t)}_{CI}$ we obtain the time series data $g_t$, as desired.

\begin{figure}
\begin{minipage}{.5\linewidth}
\centering
\begin{tikzpicture}
    \node (a1) at (-0.01,0.6) {};
    \node (v2) at (1.8,-0.2) {};
    \node (v1) at (0,0) {};
    \node (v3) at (1.25,0.3) {};
    \node (v4) at (2.4,0.3) {};

    \begin{scope}[fill opacity=0.8]

    \filldraw[fill=gray!40] ($(v1)+(-1.0,0.2)$)
        to[out=90,in=180] ($(v1)+(0,1)$)
        to[out=0,in=90] ($(v1)+(0.6,0.3)$)
        to[out=270,in=0] ($(v1)+(0,-0.6)$)
        to[out=180,in=270] ($(v1)+(-1.0,0.2)$);

   \filldraw[fill=gray!40] ($(v2)+(-1,0.2)$)
        to[out=90,in=180] ($(v2)+(0,1)$)
        to[out=0,in=90] ($(v2)+(1.2,0.3)$)
        to[out=270,in=0] ($(v2)+(0,-0.6)$)
        to[out=180,in=270] ($(v2)+(-1,0.2)$);

    \end{scope}

    \foreach \v in {1,2,3,4} {
        \fill (v\v) circle (0.1);
    }

    \fill[fill=red] (a1) circle (0.1) node [left] {$Z_{1'}$};
    \fill (v1) circle (0.1) node [left] {$X_a$} ;
    \fill[fill=red] (v2) circle (0.1) node [below] {$Z_{2'}$};
    \fill (v3) circle (0.1) node [right] {$Z_b Z_c$};
    \fill (v4) circle (0.1);

\end{tikzpicture}
\end{minipage}
\begin{minipage}{.5\linewidth}
	\begin{quantikz}
		\lstick{$\ket{0}_{1'}$} & \gate{H} & \ctrl{1} & \gate[5]{\ee^{\ii \frac{t}{2} H'}} & \ctrl{1} & \meter{} \\
		\lstick{$\ket{0}_{2'}$} & \qw  & \targ{}   &   & \targ{}  & \qw  \\
		\lstick[3]{$\rho_{abc}$} & \qw  & \qw    & \qw  & \qw  & \qw   \\
						   & \qw  & \qw    & \qw & \qw  & \qw    \\
						   & \qw  & \qw    & \qw  & \qw  & \qw \\
	\end{quantikz}
\end{minipage}%
\caption{Consider a $2$-local Hamiltonian acting on three qubits $H = X_a + Z_b Z_c$. The left-hand figure illustrates the geometric connectivity of this Hamiltonian. The qubits (black) which are geometrically connected are shown in a shared loop. It also shows the requisite geometric connectivity of the ancilla qubits (red) needed to define $H' = Z_{1'} X_a + Z_{2'} Z_b Z_c$. The right-hand figure then shows the circuit used to obtain $g_t$.}\label{fig:local-time-series}
\end{figure}
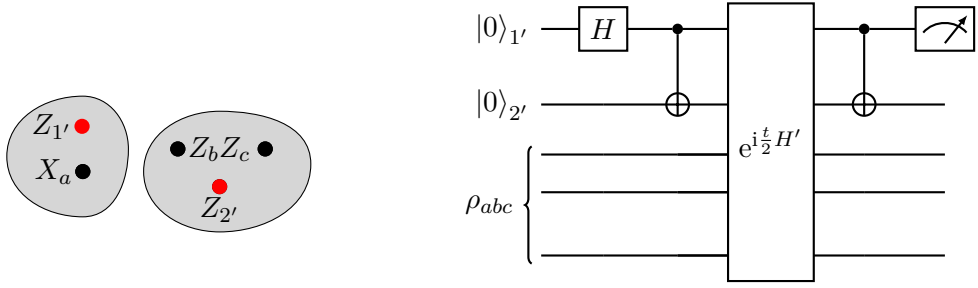

\subsection{Asymptotic Run-Time Analysis} \label{subsec:runtime}
In this section we will establish the run-time of the single most costly experiment involved in solving \nameref{prob:qeep} for $H=\sum_i \alpha_i H_i$. That is we will calculate the run-time---specified by \cref{def:time-cost}---associated with the protocol to obtain $g_T$ outlined in \cref{subsec:subcicuit-soln-QEEP}.

In previous work \cite{Clinton2020} we establish tighter bounds on the run-time of implementing $U(T) = \ee^{- \ii T H}$ where  $H = \sum_i \alpha_i h_i$ is a $k$-local Hamiltonian with $\alpha_i \leq 1$ and $\|h_i\| = 1$. First we decompose $U(T)$ via a Trotter approximation, that is set $U(T) \approx (\prod_i \ee^{\ii \delta \alpha_i h_i})^{T/\delta} + \BigO(T \delta)$, and then use the sub-circuit synthesis techniques given in \cite{Clinton2020} to implement each local Trotter step $\ee^{\ii \delta \alpha_i h_i}$ with sub-circuit run-time $\cost(\ee^{\ii \delta \alpha_i h_i})\approx \BigO(\delta^{1/(k-1)})$.\footnote{This is an abuse of notation as there are many circuits which implement a given $U$. Therefore $\cost(U)$ is not well defined, however if it is clear from context what circuit we are using to decompose $U$ we write $\cost(U)$ for readability.}

In order to perform \nameref{prob:qeep} with local control as described in the previous section we must instead implement $U'_{c}(T) = \ee^{- \ii T H'}$ where  $H' = \sum_i \alpha_i Z_{i'}  h_i$ in the manner described above. Now the run-time of each local Trotter step is $\cost(\ee^{\ii \delta \alpha_i Z_{i'}  h_i})\approx \BigO(\delta^{1/k})$. As there are $T/\delta$ Trotter steps this means
\begin{align}\label{eq:time-cost-delta}
\cost(U'_{c}(T)) \leq \BigO\left( T \delta^{\frac{1-k}{k}} \right).
\end{align}
This diverges as $\delta \rightarrow 0$. Therefore, we must choose the largest $\delta$ for a given Trotter error $\epsilon$.\footnote{Here we are take the Trotter error to be the same as the precision to which we implement \nameref{prob:qeep}.} That is we choose the value of $\delta$ which saturates the inequality $\epsilon \geq \BigO(T \delta)$. Substituting this into \cref{eq:time-cost-delta} gives
\begin{align}\label{eq:time-cost}
\cost(U'_{c}(T)) \leq \BigO \left(T^{1 + \frac{k-1}{k}} \epsilon^{-\frac{k-1}{k}} \right).
\end{align}

As discussed in \cref{subsec:altInd}, $T$ will depend in the bin-width $\eta$ and precision. We will use the bound \cref{eq:T-bound-2} to determine this, meaning
\begin{align}\label{eq:time-full}
\cost(U'_{c}(T)) \leq \BigO \left(T^{1 + \frac{k-1}{k}} \epsilon^{-\frac{k-1}{k}} \right) \quad \text{with} \quad T \rightarrow \frac{\pi}{\eta} \times \frac{\ee^{\pi \eta \epsilon/2}}{\sqrt{\ee^{\pi \eta \epsilon}-1}}.
\end{align}
This run-time does not take any connectivity assumptions into account, however we will address the issue of connectivity in our numeric benchmarking of this approach in \cref{subsec:numerics}.
As we are interested in the NISQ regime we will focus on
a numerical analysis of a specific example Hamiltonian, the $2$D spin Fermi-Hubbard Hamiltonian, rather than an asymptotic analysis of the general algorithm.

\subsection{Numerical Results for the Fermi-Hubbard Model}\label{subsec:numerics}
\begin{figure}
	\hspace*{-1.5cm}\vspace{-.5cm}
	\begin{minipage}{18cm}
		\includegraphics[width=\textwidth]{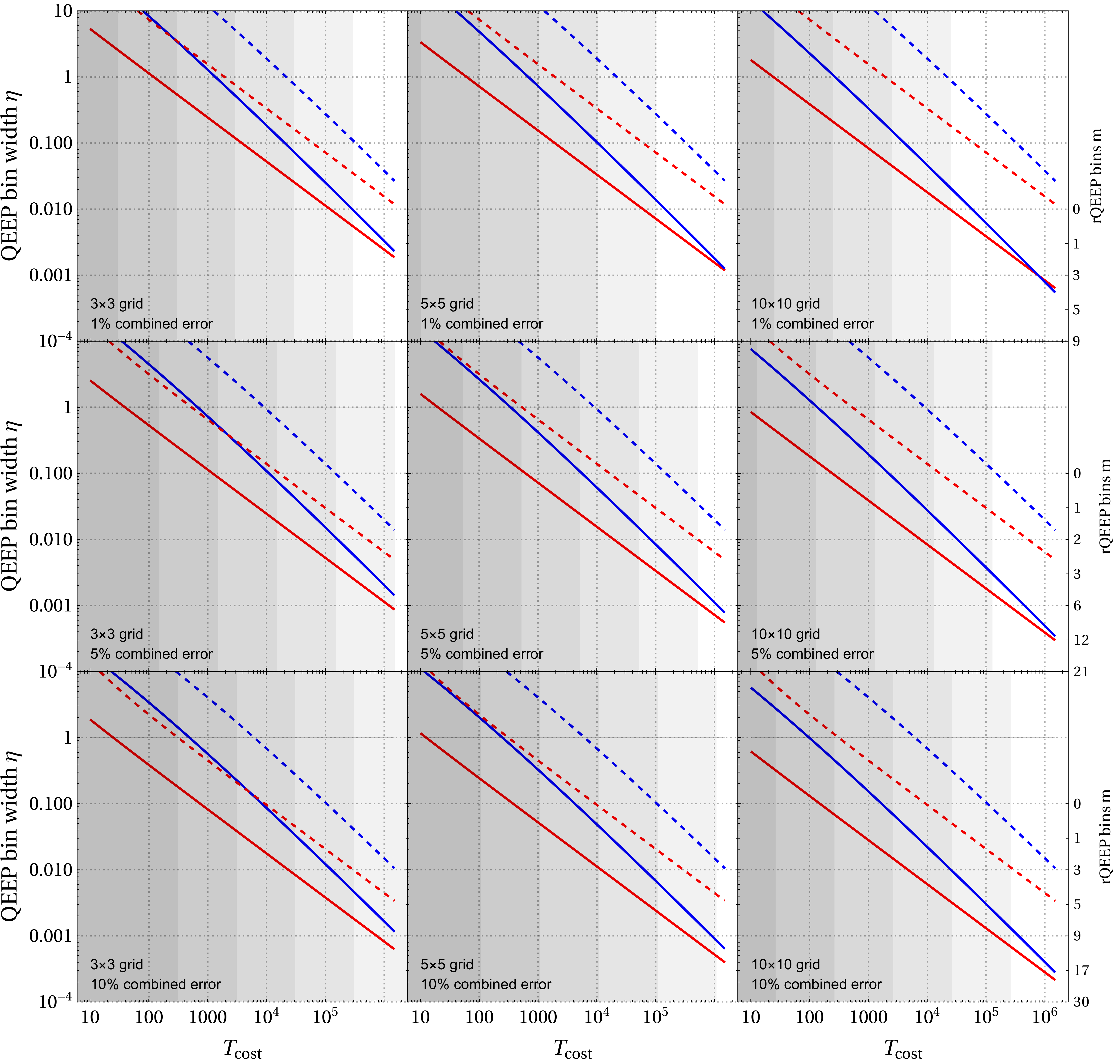}\\[1mm]
	\end{minipage}
	\caption{A comparison of the impact of using standard (blue) vs.~sub-circuit (red) synthesis methods on the achievable \nameref{prob:qeep} bin-width $\eta$ for a $3\times3$ (left column), $5\times5$ (middle column) and $10\times10$ (right column) 2D  spin Fermi-Hubbard Hamiltonian.
		The right scale alternatively shows the number of achievable \eighprob number of bins $m$, as given in \cref{eq:m}.
		Solid lines are the novel indicator function from \cref{subsec:altInd}, dashed lines are Somma's original one.
		The combined QEEP $\epsilon$ and analytic Trotter error $\epsilon_{t}$ results in a total analytic error $( \epsilon^2 + \epsilon_{t}^2)^{1/2} $.
		The shaded backgrounds indicate the amount of error due to stochastic noise; the thresholds are chosen such that they coincide with the combined Trotter and phase estimation errors (so $10$, $5$ or $1\%$ additional error due to stochastic noise, respectively). The shading indicates depolarizing noise, from $q=10^{-5}$ (darkest) to $10^{-9}$ (lightest). For $\epsilon = 0.01, \ 0.05$ and $0.1$ we find that $\Delta/ 2^{Q} = 0.16, \ 0.8$ and $1.6$ respectively.}
	\label{fig:main}
\end{figure}

We will take the following problem instance; a \EighProb input Hamiltonian which is a two-dimensional $L \times L$ lattice Fermi-Hubbard model with spin.
This Hamiltonian is made up of on-site interactions and nearest neighbour hopping terms
\begin{align}\label{eq:FH-H}
H &\coloneqq  \sum_{i=1}^{L^2} h_{\text{on-site}}^{(i)} \ + \sum_{i<j,\sigma} h_{\text{hopping}}^{(i,j,\sigma)} \\
&\coloneqq  u \sum_{i=1}^{L^2}   a^{\dagger}_{i \uparrow} a_{i \uparrow} a^{\dagger}_{i \downarrow} a_{i \downarrow}  +  v \sum_{i<j,\sigma}   \left(a^{\dagger}_{i \sigma} a_{j \sigma} + a^{\dagger}_{j\sigma} a_{i \sigma}\right)
\end{align}
which describe electrons; and with spin $\sigma =\ \uparrow$ or $\downarrow$ hopping between neighbouring sites on a lattice, with an on-site interaction between opposite-spin electrons at the same site.
We will assume that this fermionic Hamiltonian has been encoded into a qubit Hamiltonian using the compact fermion encoding \cite{derby2020}. This gives us a $4$-local qubit Hamiltonian, as the lowest on-site term (originally 2-local) and the hopping terms (originally 3-local) now need to be implemented with an extra controlling qubit.
To work out the individual gate costs, we reference \cite[Appendix B]{Clinton2020}, and for an application of either term to time $\delta$ we obtain
\begin{align}
    \text{3-local on-site:}&\quad g_3(u, \delta) = g_3(u\delta) = \frac{|u \delta|}{2} + 2 \sqrt{2 |u \delta|}\\
    \text{4-local hopping:}&\quad g_4(v, \delta) = g_4(v\delta) = 12 \sqrt[3]{2 |v \delta|}
\end{align}
for a first order Trotter formula.
We also remark that the compact encoding requires a single layer of the on-site terms, but 4 layers for the hopping terms; the total cost for a single Trotter step with time $\delta$ we have $g_3(u\delta) + 4g_4(v\delta)$.

 We must first pindown the geometric connectivity assumptions we make as these play a critical role in determining the run-time of an algorithm on a NISQ device. Here we make one of the most limiting assumptions of our analysis. If we are to avoid all SWAP overhead in the example we are about to consider, then we must assume the device consists of two stacked and connected 2D grids of qubits with nearest neighbour connectivity. This is shown in \cref{fig:connectivity}. Additionally we assume the noise model explained in \cref{subsec:error-models} applies to this device and use \cref{eq:stoch-noise-bound} to calculate the run-time region where the combined stochastic and analytic error remains bounded by $\epsilon_{tar}$ for a variety of noise parameters $q$.
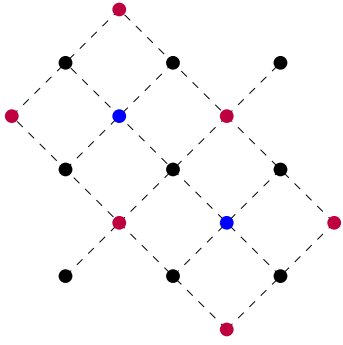
\begin{figure}
	\centering
	\begin{tikzpicture}[scale=1,rotate=45]

	\draw[dashed][step=1cm,black,very thin] (0,0) grid (2,4);
	\draw[dashed][step=1cm,black,very thin] (2,2)--(3,2);
	\draw[dashed][step=1cm,black,very thin] (-1,2)--(0,2);

	\foreach \x in {1}{
		\foreach \y in {0,2,4}{
			\node at (\x,\y)[fill=black,circle,scale=0.5]{};
		}
	}

	\foreach \x in {0,2}{
		\foreach \y in {1,3}{
			\node at (\x,\y)[fill=black,circle,scale=0.5]{};
		}
	}

	\foreach \x in {-1,3}{
		\foreach \y in {2}{
			\node at (\x,\y)[fill=black,circle,scale=0.5]{};
		}
	}

	\foreach \x in {0,2}{
		\foreach \y in {0,2,4}{
			\node at (\x,\y)[fill=purple,circle,scale=0.5]{};
		}
	}

	\foreach \x in {1}{
		\foreach \y in {1,3}{
			\node at (\x,\y)[fill=blue,circle,scale=0.5]{};
		}
	}

	\end{tikzpicture}
	\caption{The ancilla qubits and connectivity required to perform \nameref{prob:qeep} on a single spin sector of a $L \times L$ Fermi-Hubbard Hamiltonian for $L=3$. We use the compact encoding to encode a $L \times L$ Fermi-Hubbard Hamiltonian. This requires $2[L^2 + (L-1)^2 +2(L-1)]$ qubits in total. For a single spin layer as shown here, the compact encoding already requires $(L-1)^2/2$ ancilla qubits to encode the fermionic Hamiltonian, these are shown in blue. An additional $(L-1)^2/2 + 2(L-1)$ ancilla qubits to implement \nameref{prob:qeep} with local control as described in \cref{subsec:subcicuit-soln-QEEP}, these are shown in purple.}\label{fig:connectivity}
\end{figure}

Now we will proceed to analyse \nameref{prob:qeep} and then \eighprob for a specific problem instance assuming this device connectivity. We will assume that our device has a given run-time budget before it is overcome by noise. We will determine numerically how the minimum obtainable bin-width $\eta$ goes as a function of this run-time budget for various values of the target precision. That is we will perform the numerical equivalent of the asymptotic analysis of \cref{subsec:runtime} but solving \cref{eq:time-full} to get $\eta = \eta(\cost, \epsilon)$.

In order to solve \nameref{prob:qeep} with a bin-width of $\eta$ we must rescale $H$ in order to apply the bound \cref{eq:T-bound-1} or \cref{eq:T-bound-2}.
For $H$ on an $L\times L$ lattice, we know that $\|H\| \leq 5\Lambda$, where $\Lambda$ is the number of fermions in the simulation, see \cite[Sec.~E.4]{Clinton2020}; we thus need to rescale the maximum evolution time by a factor $1/10\Lambda$.
The maximum evolution time $T$ is then bounded by
\begin{align}
	\text{\cref{eq:T-bound-1}}:&\quad   T \ge \frac{\pi}{10\eta\Lambda} \times \frac{\ee^{\pi \eta \epsilon/2}}{\sqrt{\ee^{\pi \eta \epsilon}-1}}, \\
    \text{\cref{eq:T-bound-2}}:&\quad   T \ge \frac{1}{5\eta\Lambda}\left(1+\operatorname{plog}_{-1}\left( - \frac{\epsilon\eta}{32e} \right)\right)^2,
\end{align}
where we are considering two possible choices of $f(\omega)$ with each choice leading to a different bound on the maximum evolution time as derived in \Cref{app:TBound}.
We assume $U'_c(T)$ is implemented by a $p$\textsuperscript{th}-order Trotter formula using $5$ layers. Each layer is made up of mutually commuting interactions which are also disjoint. We assume the NISQ device can implement these interactions in parallel (See \cite{Clinton2020} for more details). The function $\eta = \eta(\cost, \epsilon)$ has a further dependence on which order Trotter formula is used to implement $U'_c(T)$ as this determines $\cost$. However in our numerical analysis we will simply show the minimum $\eta$ obtained by varying over $p=1, 2$ or $4$. Furthermore we will use the  numerical Trotter bounds obtained in \cite{Clinton2020}, based on extrapolating optimal Trotter step sizes from lattice sizes up to $3\times3$.

To benchmark the impact of our sub-circuit synthesis techniques we perform the same analysis but restricting ourselves to decomposing the local unitaries into CNOT gates and single qubit rotations. We cost these gates in the same framework, meaning that in this case the run-time of each local Trotter step equals the circuit depth (disregarding single-qubit gates) of the decompositions of the 3- and 4-local interactions using the lowest-depth gate decomposition method (via conjugation), which yields
\begin{align}
    \text{3-local on-site:}&\quad g'_3(u, \delta) = g'_3 = \frac{5 \pi}{4}\\
    \text{4-local hopping:}&\quad g'_4(v, \delta) = g'_4 =  \frac{5 \pi}{2}
\end{align}
and again that a single layer has a runtime $g'_3 + 4g'_4$. The results of these numerics are presented in \Cref{fig:main} and show that--within the parameter range considered--our asymptotically suboptimal choice of indicator function (solid lines) and synthesis techniques (red lines) yield improvements in the minimal achievable bin-width $\eta$ for a given depolarizing noise rate and total runtime $\cost$.

\section{Discussion and Conclusions}\label{sec:conclusions}
In this work we have investigated three different avenues of research concerning quantum phase estimation.
We recap and discuss them seperately in the following.

\paragraph{Subcircuit Analysis.} We have built on our previous work, recapped in \Cref{sec:prelim}, where we defined a non-standard circuit error model from an abstraction of the capabilities of NISQ hardware and then introduced gate synthesis techniques tailored to this framework. Here we have applied this framework to the problem of \nameref{prob:qeep}---introduced in \cite{Somma2019}---to determine whether our gate synthesis techniques yield any improvements within this circuit error model. 

In order to do this we have outlined a simple implementation of \nameref{prob:qeep} in \Cref{sec:subcircuit-QEEA} aimed at minimizing runtime by performing controlled unitary evolution using only subcircuit gates. This \nameref{prob:qeep} implementation preserves any potential subcircuit gains, shown in \cite{Clinton2020}, by exchanging the global control ancilla of the \nameref{prob:qeep} implementation in \cite{Somma2019} with local control ancillas and the creation of a global CAT state, in order to
highlight its runtime minimizing properties within our
minimize the runtime in the per-time error model. 

It is essentially true by the definition of our per-time error model that our sub-circuit synthesis techniques will have an \emph{asymptotic} advantage over standard synthesis techniques in this framework. Furthermore, as this per-time error model is targetting NISQ-era algorithms it would also be misleading to confine ourselves the simple asymptotic analysis of \Cref{subsec:runtime}. For both of these reasons we have anchored our conclusions in a numerical analysis of the 2D spin Fermi-Hubbard model. We have considered how implementing \nameref{prob:qeep} with local control relates to the connectivity of a NISQ device (see \Cref{fig:connectivity}) and then perform numerics to obtain the results shown in \Cref{fig:main}. These plots demonstrate a numeric advantage for our synthesis techniques within a per-time error model. Modulo the applicability of a per-time error model, these techniques may allow one to perform \nameref{prob:qeep} on an $L=3$ spin Fermi-Hubbard model with a bin width of $\eta=0.1$ and $\epsilon=0.1$ if given access to a 3D architecture with the connectivity of \Cref{fig:connectivity} with a depolarizing noise rate of $q=10^{-6}$. In comparison, we estimate standard synthesis techniques would require  $q=10^{-7}$. We emphasize that this estimate is equivalent to a standard circuit error model, as all the gates in this comparison are costed as $\BigO(1)$ in our framework as explained in \Cref{sec:prelim}.

This encapsulates the first avenue of research explored in this paper, a development and application of our work in \cite{Clinton2020} to a assess the possibility of performing \nameref{prob:qeep} in the NISQ-era.

\paragraph{Reformulating \nameref{prob:qeep}.}
Secondly, we have reframed the \nameref{prob:qeep} protocol presented in \cite{Somma2019} as an algorithm which allows one to sample from a convolution between the spectrum of a Hamiltonian, $G(\omega)$ and a freely chosen function $f(\omega)$.\footnote{Freely chosen from functions which satisfy the conditions explained in \Cref{subsec:qeep}} This is explained in \Cref{subsec:qeep} and specifically via \Cref{eq:qeep-as-conv}.  By recasting the problem and algorithm in this way we hope to clarify which aspects we can modify and adapt to different computational problems.
	
Most trivially our analysis decouples the precision parameter $\epsilon$ from the bin-width $\eta$. However we have also explored the possibility of tailoring our choice of $f(\omega)$ to problem and parameter regime of interest, in our case an instance of \nameref{prob:qeep} with a Fermi-Hubbard Hamiltonian in a NISQ regime.
	
We define one such alternative indicator function in \Cref{app:TBound} and compare the resultant lower bound on the maximum evolution time $T$ with the original indicator function of \cite{Somma2019} in \Cref{fig:T-bound}. We emphasize that the original function gives asymptotically superior results, as it was intended to. Despite this, within the non-asymptotic parameter range of \Cref{fig:main} we can see that for both standard and subcircuit synthesis techniques our bound gives improvements in the minimum obtainable $\eta$. For example considering standard gate synthesis techniques, for an $L=3$ spin Fermi-Hubbard model \Cref{fig:main} shows that for $\epsilon =0.1$ and for a target bin width of $\eta = 0.1$, our alternative indicator function improves the noise requirements by an order of magnitude. 

As we see from the above, the optimal choice of indicator function may depend on the task at hand. It is also possible that other applications might benefit from taking non-uniform samples from the time series $g(t)$ and employing more sophisticated classical signalling processing methods to reconstruct $p(\omega) = [G*f](\omega)$. By formulating \nameref{prob:qeep} as a protocol for sampling from a $p(\omega)$ we hope to make it easier to explore these questions in future work. 

\paragraph{rQEEP.}
Finally, we suggest an application of \nameref{prob:qeep} having thus reformulated it as a general tool. We consider a problem analogous to matlab's \texttt{eig} or Mathematica's \texttt{Eigenvalues}. This was outlined in \Cref{subsec:generic-QEEA-def} and in \Cref{subsec:Generic-QEEA-alg} we demonstrate how to solve it using \nameref{prob:qeep} as a subroutine. We also incorporate this problem into our numerics as shown in \Cref{fig:main} and find that it is more demanding than performing \nameref{prob:qeep} alone---as expected, due to the superior guarantees \eighprob gives on the placement of eigenvalues within each bin.

\newpage

\printbibliography

@book{Kittel2004,
  asin = {047141526X},
  author = {Kittel, Charles},
  description = {Amazon.com: Introduction to Solid State Physics (9780471415268): Charles Kittel: Books},
  dewey = {530.41},
  ean = {9780471415268},
  edition = 8,
  isbn = {9780471415268},
  publisher = {Wiley},
  title = {Introduction to Solid State Physics},
  url = {http://www.amazon.com/Introduction-Solid-Physics-Charles-Kittel/dp/047141526X/ref=dp_ob_title_bk},
  year = 2004
}

@article{review_mott,
  title = {Metal-insulator transitions},
  author = {Imada, Masatoshi and Fujimori, Atsushi and Tokura, Yoshinori},
  journal = {Rev. Mod. Phys.},
  volume = {70},
  issue = {4},
  pages = {1039--1263},
  numpages = {0},
  year = {1998},
  month = {Oct},
  publisher = {American Physical Society},
  doi = {10.1103/RevModPhys.70.1039},
  url = {https://link.aps.org/doi/10.1103/RevModPhys.70.1039}
}

@article{Brown2011,
   title={Computational Difficulty of Computing the Density of States},
   volume={107},
   ISSN={1079-7114},
   url={http://dx.doi.org/10.1103/PhysRevLett.107.040501},
   DOI={10.1103/physrevlett.107.040501},
   number={4},
   journal={Physical Review Letters},
   publisher={American Physical Society (APS)},
   author={Brown, Brielin and Flammia, Steven T. and Schuch, Norbert},
   year={2011},
   month={Jul}
}

@article{Kitaev1995,
    author = "Kitaev, A. Yu.",
    title = "{Quantum measurements and the Abelian stabilizer problem}",
    eprint = "quant-ph/9511026",
    archivePrefix = "arXiv",
    month = "11",
    year = "1995"
}

@article{IBMshortpulse,
  title = {Simulating the dynamics of braiding of Majorana zero modes using an IBM quantum computer},
  author = {Stenger, John P. T. and Bronn, Nicholas T. and Egger, Daniel J. and Pekker, David},
  journal = {Phys. Rev. Research},
  volume = {3},
  issue = {3},
  pages = {033171},
  numpages = {11},
  year = {2021},
  month = {Aug},
  publisher = {American Physical Society},
  doi = {10.1103/PhysRevResearch.3.033171},
  url = {https://link.aps.org/doi/10.1103/PhysRevResearch.3.033171}
}

@article{Shor,
	title={Polynomial-Time Algorithms for Prime Factorization and Discrete Logarithms on a Quantum Computer},
	volume={26},
	ISSN={1095-7111},
	url={http://dx.doi.org/10.1137/S0097539795293172},
	DOI={10.1137/s0097539795293172},
	number={5},
	journal={SIAM Journal on Computing},
	publisher={Society for Industrial & Applied Mathematics (SIAM)},
	author={Shor, Peter W.},
	year={1997},
	month={Oct},
	pages={1484–1509}}

@article{derby2020,
  title = {Compact fermion to qubit mappings},
  author = {Derby, Charles and Klassen, Joel and Bausch, Johannes and Cubitt, Toby},
  journal = {Phys. Rev. B},
  volume = {104},
  issue = {3},
  pages = {035118},
  numpages = {12},
  year = {2021},
  month = {Jul},
  publisher = {American Physical Society},
  doi = {10.1103/PhysRevB.104.035118},
  url = {https://link.aps.org/doi/10.1103/PhysRevB.104.035118}
}

@article{Clinton2020,
  title={Hamiltonian simulation algorithms for near-term quantum hardware},
  author={Clinton, Laura and Bausch, Johannes and Cubitt, Toby},
  journal={Nature communications},
  volume={12},
  number={1},
  pages={1--10},
  year={2021},
  publisher={Nature Publishing Group},
doi={https://doi.org/10.1038/s41467-021-25196-0}
}

@article{Wocjan2006,
	author = {Wocjan, Pawel and Zhang, Shengyu},
	year = {2006},
	month = {07},
	pages = {},
	title = {Several natural BQP-Complete problems},
  	eprint="quant-ph/0606179",
    	archivePrefix={arXiv},
    	primaryClass={quant-ph}
}

@article{Somma2019,
author = {Somma, Rolando},
year = {2019},
month = {11},
pages = {},
title = {Quantum eigenvalue estimation via time series analysis},
volume = {21},
journal = {New Journal of Physics},
doi = {10.1088/1367-2630/ab5c60}
}

@article{Obrien2019,
author = {O'Brien, Thomas and Tarasinski, Brian and Terhal, Barbara},
year = {2019},
month = {01},
pages = {},
title = {Quantum phase estimation of multiple eigenvalues for small-scale (noisy) experiments},
volume = {21},
journal = {New Journal of Physics},
doi = {10.1088/1367-2630/aafb8e}
}

@article{atia2017fast,
  title={Fast-forwarding of Hamiltonians and exponentially precise measurements},
  author={Atia, Yosi and Aharonov, Dorit},
  journal={Nature communications},
  volume={8},
  number={1},
  pages={1--9},
  year={2017},
  publisher={Nature Publishing Group},
doi = {https://doi.org/10.1038/s41467-017-01637-7}
}

@misc{lin2021heisenberg,
      title={Heisenberg-limited ground state energy estimation for early fault-tolerant quantum computers}, 
      author={Lin Lin and Yu Tong},
      year={2021},
      eprint={2102.11340},
      archivePrefix={arXiv},
      primaryClass={quant-ph}
}

\newpage

\appendix
\section{Run-time of rQEEP}\label{ap:rQEEP-proofs}
Here we prove \Cref{thm:runtime-ap}. We first note that if $y\lwr_i = y\upr_i$ for all $i$---which means that no eigenvalues were detected to lie in the shaded region shown for a single bin in \Cref{fig:rQEEP-envelopes}---the algorithm would return an answer to the \EighProb.
The issue is that we cannot be certain that $H$ has no eigenvalues in these regions; and furthermore, since the $q_i$ are \emph{approximate} (since they include Fourier truncation and measurement errors) it is unlikely that any of them will be precisely zero.

However, we can estimate with what likelihood any of $H$'s eigenvalues lie these regions.
As these ambiguous shaded regions are centred around the $x_i$'s, we can derive an expected number of eigenvalues that will lie within the ambiguous region. We use the following lemma:
\begin{lemma}\label{lem:exp}
	Let the setup be as in the \eighprob algorithm outlined in \Cref{sec:generic-QEEA} but given an ideal \nameref{prob:qeep} solver with $\epsilon=0$ or equivalently $\mathbf{p}=\mathbf{q}$.
	Then
	\begin{align}
		\Expectation\left[\sum_i(y\upr_i - y\lwr_i)\right] \le 2^Q \times 6\eta(m+1)^2.
	\end{align}
\end{lemma}
\begin{proof}
	It is easy to see that any ambiguous region between the two envelopes $R\lwr_i$ and $R\lwr_{i+1}$ has width $2\eta$, and hence leaves room for up to two bin centers $w_j$, let us call them $w_a$ and $w_b$.
	Neither of the $q_a,q_b$ will be counted in $y\lwr_i,y\lwr_{i+1}$, but both will be counted in $y\upr_i,y\upr_{i+1}$.
	This means that the ambiguous region around $w_a,w_b$ has width $3\eta$ (left support end of $f(w_a)$ till right support end of $f(w_b)$).
	Noting that overlapping ambiguous regions (e.g.\ when $x_i$ and $x_{i+1}$ are very close) just reduce the overall ambiguous area, the remainder of the proof is a volume argument: if we consider a \emph{fixed} index $i$ and look at the ambiguous regions to the right (associated with $x_{i}$) and the left (associated with $x_{i-1}$), the position of right region  is chosen uniformly at random from $S_{i}$ and  position of left region is chosen uniformly at random from $S_{i-1}$. For each eigenvalue and ambiguous region we can define a function, $\chi_{i,\lambda}(x_{i}) \in \{1,0\}$, which indicates whether the ambiguous region associated with $x_i$ contains the eigenvalue $\lambda$. Now we consider the random variables $\chi_{\lambda}(x_{i})$ and $\chi_{\lambda}(x_{i - 1})$ where we have made the $i$ subscript implicit through the variable $x_i$. As each ambiguous region is chosen from a segment of width $\xi/2$ the probability that $\chi_{\lambda}(x_{i})=1$ is upper bounded by $6 \eta/\xi$ and likewise for $\chi_{\lambda}(x_{i - 1})$. As expectation values are linear even for correlated random variables we can say that
	\begin{align}
		\Expectation\left[\chi_{\lambda}(x_{i - 1})\right] + \Expectation\left[\chi_{\lambda}(x_{i})\right] =\Expectation\left[\chi_{\lambda}(x_{i - 1}) + \chi_{\lambda}(x_{i})\right] \leq \frac{12 \eta}{\xi},
	\end{align}
	and by the same argument that
	\begin{align}
		\Expectation\left[n\upr_i - n\lwr_i\right] = \Expectation\left[\sum_{\lambda} \left(\chi_{\lambda}(x_{i - 1}) + \chi_{\lambda}(x_{i})\right)\right] \leq 2^Q \times \frac{12 \eta}{\xi}.
	\end{align}
	Jointly, the position of the ambiguous intervals is of course not uncorrelated, since we choose each within a constrained segment; but again we use the fact that expectation values are linear, and so we have
	\begin{align}
		\sum_{i=1}^{m+1} \Expectation\left[y\upr_i - y\lwr_i\right] =  \Expectation\left[\sum_{i=1}^{m+1}\left(y\upr_i - y\lwr_i\right)\right] \leq 2^Q \times 6\eta(m+1)^2
	\end{align}
	The claim follows.
\end{proof}
\begin{theorem}\label{thm:runtime-ap}
	If assume access to a \nameref{prob:qeep} solver with quantum runtime $\cost(Q,\eta,\epsilon,c)$ then the following holds. For a point spacing $\xi > 0$ (or equivalently $m := 2/\xi$), deviation $\Delta\ge0$, and confidence $c$, the overall quantum runtime to answer \eighprob using \eighprob is 
	\begin{align}
		\cost\left(Q,\frac{1}{12 (m+1)^2} \times \frac{\Delta}{2^Q},\frac{1}{16} \times \frac{\Delta}{2^Q},c \right),
	\end{align}
	and the additional classical overhead is $\log_{(1-\frac{c}{2})}(1-c)$.
\end{theorem}
\begin{proof}
To prove that the algorithm does indeed answer \nameref{prob:sQEEA} is now straightforward: in every run of the subloop, \nameref{prob:qeep} succeeds with probability $c$ to produce the desired probability vector $\mathbf q$ which satisfies $|\mathbf{p}- \mathbf{q}|  \leq \epsilon$. Then if $y'_i$ is constructed from elements of $\mathbf{p}$ and $y_i$ denotes the actual output of the algorithm and is constructed from elements of $\mathbf{q}$ we have
\begin{align}
	\sum_i |y_i - n_i| &\leq \sum_i |y_i - y_i'| + \sum_i |y'_i - n_i| .
\end{align}
We first consider the second term. The envelope method for the chosen bin width produces an estimate $y'_i$ of $n_i \in [y\lwr_i, y\upr_i]$. Then using Markov's inequality and \Cref{lem:exp} we can say that with probability $\ge 1/2$ and that
\begin{align}
	\sum_i  \left(y\upr_i - y\lwr_i\right) < 2 \times \Expectation\left[\sum_i \left(y\upr_i - y\lwr_i \right)\right],
\end{align}
and therefore that with probability $\geq 1/2$ that
\begin{align}
	\sum_i |y'_i - n_i| &= \sum_i \left| \frac{y\upr_i+y\lwr_i}{2} - n_i \right| \\
	&\le \frac12 \sum_i | y\upr_i - y\lwr_i | \\
	&\le 2^Q \times 6(m+1)^2 \times \frac{\Delta}{2^{Q+1} 6 (m+1)^2}
	=\frac\Delta2.
\end{align}
We now consider the first term. From Somma's analysis we have that $\forall j \ |q_j - p_j| \leq \epsilon/(M+1)$, and by definition $M = \lfloor 1/\eta \rfloor$ so we have
\begin{align}
	\sum_{i=1}^{m+1} |y_i - y_i'| &\leq 2^Q \sum_i \sum_{\omega_j \in R_i^{\upr}}|q_j - p_j| \\
	&\leq 2^Q  \sum_{i, \omega_j } \frac{\epsilon}{M+1} \leq 2^Q \sum_{i,  \omega_j } \frac{\epsilon}{M} .
\end{align}
By definition $|x_i - x_{i-1}| \leq \xi$. The width of $R_i^{\text{upr}}$ is then at most $\xi +2 \eta$. If we assume that $\xi/\eta > 2$ then there are at most $(\xi +2 \eta)/\eta \leq 2\xi/\eta$  bins labelled by $\omega_j$ in $R_i^{\upr}$ and therefore
\begin{align}
	\sum_{i=1}^{m+1} |y_i - y_i'| &\leq 2^{Q} \sum_{i}  \left( \frac{ 2 \xi}{ \eta } \times \frac{\epsilon}{M} \right) = 2^{Q+1} \sum_{i}  \xi \epsilon.
\end{align}
As $\xi = 2/m$ and $m+1 < 2m$ this simplifies to
\begin{align}
	\sum_{i=1}^{m+1} |y_i - y_i'|	&\leq  2^{Q+1} \times 2 m \times \xi \epsilon  \\
	&=  4 \times 2^{Q+1}  \epsilon
	=\frac{\Delta}{2}.
\end{align}
\emph{If} this bound succeeded, the overall deviation due to the envelope, plus the \nameref{prob:qeep} deviation due to the precision parameter $\epsilon=\Delta /(8 \times 2^{Q+1})$, thus leads to a deviation $\sum_j |y_j - n_j| \le \Delta$.

In order to ensure that the overall procedure succeeds with confidence at least $c$, we note that the success probability of one run of the subloop is now $c/ 2$ (the $c$ from \nameref{prob:qeep}, the $1/2$ from the deviation in expectation of the envelope method). From the opposite perspective, the likelihood of failure of all rounds is thus
\begin{align}
	\left(1-\frac{c}{2}\right)^{\lceil \log_{(1-\frac{c}{2})}(1-c) \rceil}
	\le 1-c,
\end{align}
and hence the overall success probability is lower-bounded by the desired confidence level $c$. The theorem statement then follows.
\end{proof}

\section{QEEP Indicator Functions}\label{app:TBound}
\subsection{Alternative Indicator Function} \label{subsec:altInd}
Computing an explicit lower bound for $T$ using Somma's original indicator function is a non-trivial task. We include an asymptotic calculation in \cref{ap:sommaT}. However for near term applications, a simpler indicator function may yield better analytic bounds in the non-asymptotic regime. For this reason we consider a different indicator function, which allows us to derive an explicit bound on $T$ without further restrictions when said bound holds.
To this end---and re-using the symbols $f_j$ and $F_j$ in the following---we define
\begin{equation}\label{eq:f_j2}
	f(w) \coloneqq \begin{cases}
		\cos^2\left(\frac{\pi w}{2\eta}\right)  &  -\eta < w < \eta \\ 0 & \text{otherwise}.
	\end{cases}
\end{equation}
It is easy to check that these indicator functions satisfy the conditions of \nameref{prob:qeep}.
Furthermore, we can explicitly derive its Fourier transform (and thus the Fourier series coefficients) from \cref{eq:f-F-pair}, namely
\begin{align}
	F(t) = \begin{cases}
		\displaystyle\frac{\pi \sin(t\eta)}{2 t \pi^2 - 2 t^3 \eta^2} & t \neq 0 \\[5mm] \displaystyle\frac{\eta}{2\pi} & \text{otherwise}
	\end{cases}
\end{align}
and thus for $t>\pi/\eta$ we have
\begin{align}
	|F(t)| \le \frac{1}{2\pi} \times \frac{1}{t^3\eta^2/\pi^2 - t}.
\end{align}
Together with \cref{eq:pApprox} we then get
\begin{align}
	|p_j - p'_j| \le \frac{1}{2\pi} \sum_{|t|>T} \frac{\pi^2}{t^3\eta^2 - t \pi^2} \le \frac{1}{2\pi} \int_{T}^\infty \frac{\pi^2}{t^3\eta^2 - t \pi^2} \mathrm dk
	= \frac{1}{4\pi} \log \left(\frac{\eta ^2 T^2}{\eta ^2 T^2-\pi ^2}\right)
	\overset{!}{\le} \frac{\epsilon}{2(M+1)}
\end{align}
which together with $M+1 \leq 2/\eta$ yields
\begin{equation}\label{eq:T-bound-2}
	T \ge \frac{\pi}{\eta} \times \frac{\ee^{\pi \eta \epsilon/2}}{\sqrt{\ee^{\pi \eta \epsilon}-1}}.
\end{equation}
Both for \cref{eq:T-bound-1,eq:T-bound-2} we gain a factor of $M$ in case we are interested in the max-norm instead of the 1-norm, as mentioned.
We compare this new bound with \citeauthor{Somma2019}'s bound in \cref{fig:T-bound}.
\begin{figure}[t]
	\hspace*{-1.5cm}
	\begin{minipage}{18cm}
		\includegraphics[width=6cm]{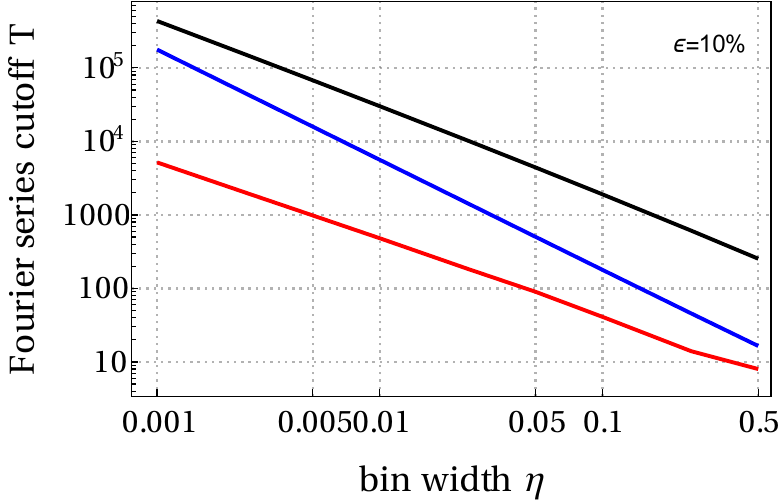}
		\includegraphics[width=6cm]{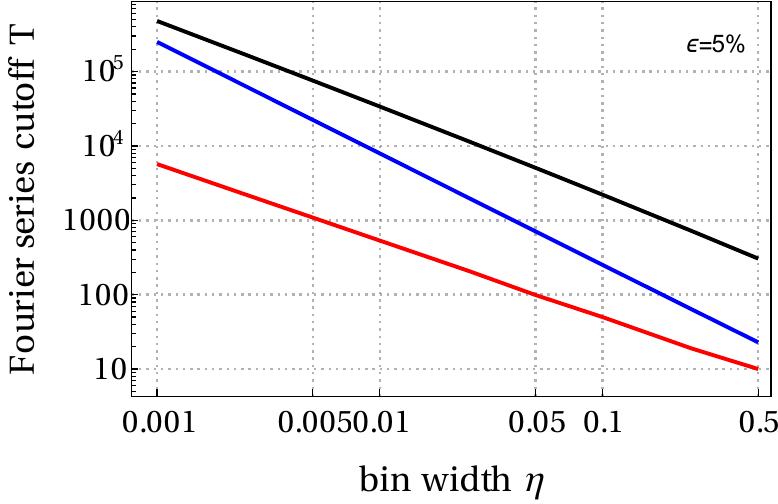}
		\includegraphics[width=6cm]{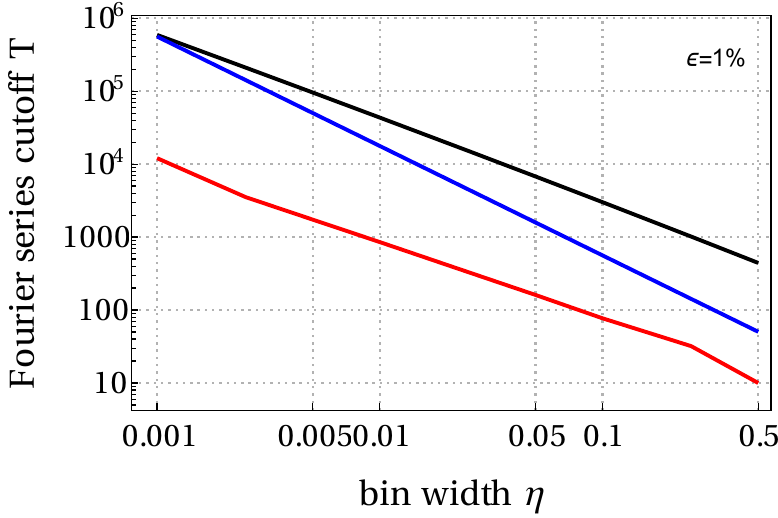}
	\end{minipage}
	\caption{Maximum evolution time $T$ for a given desired bin width $\eta$, for $\epsilon=0.1$ (left), $\epsilon=0.05$ (middle), and $\epsilon=0.01$ (right).
		Plotted are \citeauthor{Somma2019}'s bound (grey, top line) from \cref{eq:T-bound-1}, our new bound (blue, middle line) from \cref{eq:T-bound-2}, and a numerical bound obtained by summing the absolute values of the first $10^5$ Fourier coefficients in \cref{eq:fourier-bound1} up exactly, using the spectrum for a $3\times3$ spinless FH model, and using \cref{eq:T-bound-2} for the coefficients $>10^5$.}
	
	\label{fig:T-bound}
\end{figure}

\subsection{Original Indicator Function}\label{ap:sommaT}
In this section we find a lower bound on the Fourier series cutoff $T$ in \cref{eq:pApprox}, when using Somma's indicator function, as described in  \cref{subsec:qeep}. We wish to determine the minimal $T$ such that
\begin{align}
\|\mathbf{p} - \mathbf{p}'\| \leq \frac{\epsilon}{2}.
\end{align}
where $\mathbf{p}$ is given in Equation \ref{eq:pExact}. By \cref{eq:pApprox,eq:pExact}, we have
\begin{align}\label{eq:fourier-bound1}
|p_j - p'_j| \leq  \sum_{|t| > T} |F_t^j| |g_t|.
\end{align}
Since $|g_t| \leq 1$, it suffices to find a bound $F_t^j$. Assuming that $T$ is chosen large enough such that Equation \ref{eq:Fdecay} holds gives us
\begin{equation}
|p_j - p'_j| \leq  \sum_{|t|>T} \eta \left|\exp\left(- \sqrt{|t| \eta/2 }\right)\right| \le 2 \eta \int_{t=T}^\infty \exp\left(- \sqrt{|T| \eta/2 }\right) \mathrm dt = 8 \ee^{-\sqrt{T\eta/2}}\left(1 + \sqrt{T \eta/2} \right)
\end{equation}
Distributing $\epsilon/2$ equally amongst the $M+1 \leq 2M$ bins gives the bound on the Fourier truncation $T$ for $T\ge \alpha/\eta$
\begin{equation}\label{eq:T-bound-1}
    \ee^{-\sqrt{T\eta/2}}\left(1 + \sqrt{T \eta/2} \right) \overset{!}{\le} \frac{\epsilon}{32M}
    \quad\Longrightarrow\quad
    T \overset!\ge \frac2\eta \left( 1 + \operatorname{plog}_{-1}\left(-\frac{\epsilon\eta}{32e} \right) \right)^2.
\end{equation}
Here the product logarithm $\operatorname{plog}_{-1}(y)$ denotes the $-1$\textsuperscript{th} solution for $w$ in the equation $y=w \ee^w$.
We emphasise that this is an asymptotic result, finding $\alpha$ is not straightforward.

\end{document}
